\documentclass[a4paper,USenglish,cleveref, autoref, thm-restate]{lipics-v2021}
\hideLIPIcs  %uncomment to remove references to LIPIcs series (logo, DOI, ...), e.g. when preparing a pre-final version to be uploaded to arXiv or another public repository

\usepackage{tikz}
\usepackage{algorithm}
\usepackage{booktabs}
\usepackage{algpseudocode}
\usepackage{empheq}

\newcounter{assump}

\newcommand{\opt}{\mathrm{opt}}
\newcommand{\alg}{\mathrm{alg}}

\title{Primal-Dual Online Algorithms for the Parking Permit Problem}

\author{Christian Coester}{Department of Computer Science, University of Oxford, United Kingdom}{christian.coester@cs.ox.ac.uk}{https://orcid.org/0000-0003-3744-0977}{}%TODO mandatory, please use full name; only 1 author per \author macro; first two parameters are mandatory, other parameters can be empty. Please provide at least the name of the affiliation and the country. The full address is optional. Use additional curly braces to indicate the correct name splitting when the last name consists of multiple name parts.

\author{Alex Turoczy}{Department of Computer Science, University of Oxford, United Kingdom}{alexanderturoczy@gmail.com}{https://orcid.org/0009-0003-1768-7427}{}

\authorrunning{C. Coester and A. Turoczy} %TODO mandatory. First: Use abbreviated first/middle names. Second (only in severe cases): Use first author plus 'et al.'

\Copyright{Christian Coester and Alex Turoczy} %TODO mandatory, please use full first names. LIPIcs license is "CC-BY";  http://creativecommons.org/licenses/by/3.0/

 %TODO mandatory: Please choose ACM 2012 classifications from https://dl.acm.org/ccs/ccs_flat.cfm 


\funding{Funded by the European Union (ERC, CCOO, 101165139). Views and opinions expressed are however those of the author(s) only and do not necessarily reflect those of the European Union or the European Research Council. Neither the European Union nor the granting authority can be held responsible for them.}%optional, to capture a funding statement, which applies to all authors. Please enter author specific funding statements as fifth argument of the \author macro.

\nolinenumbers %uncomment to disable line numbering

\EventEditors{Philip Bille, Seth Pettie, and Sabine Storandt}
\EventNoEds{3}
\EventLongTitle{34th Annual European Symposium on Algorithms (ESA 2026)}
\EventShortTitle{ESA 2026}
\EventAcronym{ESA}
\EventYear{2026}
\EventDate{August 31--September 4, 2026}
\EventLocation{L'Aquila, Italy}
\EventLogo{}
\SeriesVolume{388}
\ArticleNo{79}
\begin{document}

\maketitle

%TODO mandatory: add short abstract of the document
\begin{abstract}
The Parking Permit Problem (PPP), first studied by Meyerson, is a classic online problem generalizing the ski rental problem. We re-examine the PPP using the primal-dual scheme, obtaining simple algorithms with superior performance guarantees. Unlike previous work, which relied on reductions that degraded competitive ratios, we work with the problem's structure directly. We also provide near-matching lower bounds. Using the primal-dual framework, we find the PPP's deterministic competitive ratio exactly, and the randomized competitive ratio within an \emph{additive} constant. 
\end{abstract}

\section{Introduction}

The Parking Permit Problem (PPP) is a fundamental problem in online algorithms. It is the seminal \emph{leasing} problem, meaning an algorithm may purchase a resource to satisfy constraints, but where the purchases last a fixed contiguous time duration, regardless of whether it is used in this time. In this paper, we reinvestigate this classic problem via the primal-dual framework, obtaining new and simpler algorithms with improved performance guarantees.%that are more efficient, yet simpler.

The PPP can be defined by the following canonical parking permit scenario. Suppose that you walk to work when the weather is good, and drive when it rains. When you drive, you must possess a valid parking permit, but there are $K$ different types. Each permit type lasts a different duration and expires on a fixed day regardless of how many times it is utilized, where longer permits are cheaper on a per-day basis. The dilemma in the PPP is: what permit purchases should you make?

%We evaluate our algorithms according to a standard measure in the field of online algorithms, by its \emph{competitive ratio}. An algorithm's competitive ratio for a given online instance is the ratio between its cost, and the cost of the optimal solution had the online information been known from the beginning. An algorithm's competitive ratio is the supremum of its competitive ratio over problem instances, and the competitive ratio of a problem is the infimum over all online algorithms. 

% Prior work established the competitive ratio of the PPP to be $\Theta(K)$ and $\Theta(\log K)$ for deterministic and randomized algorithms, respectively~\cite{ppporiginal}. These bounds left relatively large constant factor gaps between the upper and lower bounds of $24$ (deterministic) and $32\ln 2\approx 22.18$ (randomized).

Essentially all prior work on PPP and variants \cite{ppporiginal, multippp} involve a reduction to what we shall call the \emph{Laminar} PPP. In the Laminar PPP, the permit cannot just be purchased starting from any arbitrary day. Instead, the permit durations form a laminar set family, meaning that two different permits' durations either do not overlap, or one is contained inside the other (see Figure~\ref{fig:laminar-diagram}). The reduction to the Laminar PPP loses a constant factor in the competitive ratio. In this paper, we provide techniques for solving leasing problems without assumptions of laminarity, leading to more direct algorithms with superior performance.
\begin{figure}[ht]
    \centering
    \begin{tikzpicture}[x=0.7cm,y=0.7cm, line cap=round]
        % --- Axis with ticks ---
        \draw[thick,-latex] (0.8,0) -- (13.8,0);
        \node[right] at (13.8,0) {Time};
        \foreach \i in {1,...,13} {
          \draw (\i,0) -- (\i,0.35);
        }
        \node[below] at (1.5,0) {1};
        \node[below] at (12.5,0) {12};
        % Positions for the right-side labels
        \def\xlabel{13.6}
        \def\yone{0.8}
        \def\ytwo{1.6}
        \def\ythree{2.4}
        \def\yfour{3.2}
        % --- Row 4: solid bar ---
        \draw[line width=2pt] (1.1,\yfour) -- (12.9,\yfour);
        \node[right] at (\xlabel,\yfour) {4};
        % --- Row 3: three long segments ---
        \draw[line width=2pt] (1.1,\ythree) -- (4.9,\ythree);
        \draw[line width=2pt] (5.1,\ythree) -- (8.9,\ythree);
        \draw[line width=2pt] (9.1,\ythree) -- (12.9,\ythree);
        \node[right] at (\xlabel,\ythree) {3};
        % --- Row 2: several medium segments ---
        \draw[line width=2pt] (1.1,\ytwo) -- (2.9,\ytwo);
        \draw[line width=2pt] (3.1,\ytwo) -- (4.9,\ytwo);
        \draw[line width=2pt] (5.1,\ytwo) -- (6.9,\ytwo);
        \draw[line width=2pt] (7.1,\ytwo) -- (8.9,\ytwo);
        \draw[line width=2pt] (9.1,\ytwo) -- (10.9,\ytwo);
        \draw[line width=2pt] (11.1,\ytwo) -- (12.9,\ytwo);
        \node[right] at (\xlabel,\ytwo) {2};
        % --- Row 1: many short segments ---
        \foreach \x in {1.1,2.1,3.1,4.1,5.1,6.1,7.1,8.1,9.1,10.1,11.1,12.1} {
          \draw[line width=2pt] (\x,\yone) -- (\x+0.8,\yone);
        }
        \node[right] at (\xlabel,\yone) {1};
    \end{tikzpicture}
    \caption{Illustration of the Laminar PPP with permit durations of length $1$, $2$, $4$ and $12$.}
    \label{fig:laminar-diagram}
\end{figure}
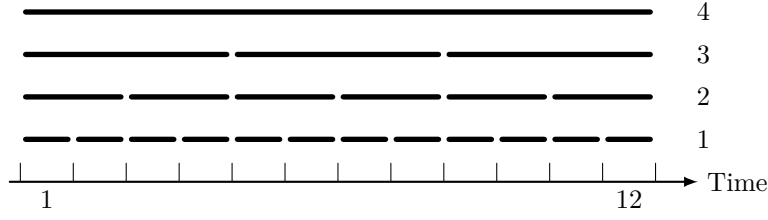

\subsection{Our Results}

Using primal-dual algorithms, we prove new upper and lower bounds for both deterministic and randomized algorithms. Our results in comparison to previous work are summarized in Table \ref{tab:ppp-bounds}.

\begin{table}[ht]
\centering
\small
\begin{tabular}{llcc}
    \toprule
    & & Deterministic PPP & Randomized PPP \\
    \midrule
    \multirow{2}{*}{Upper bounds}
        & Previous & $4K$ \cite{ppporiginal} & $16 \ln K$ \cite{ppporiginal} \\
        & New      & $K$ & $\ln K + \ln \ln K + O(1)$ \\
    \midrule
    \multirow{2}{*}{Lower bounds}
        & Previous & $K/6$ \cite{ppporiginal}, $\ 3$ for $K=3$ \cite{threeppp} & $\log_2(K)/2$ \cite{ppporiginal} \\
        & New      & $K$ & $\ln K + \ln \ln K + o(1)$ \\
    \bottomrule
\end{tabular}
\vspace{6pt}
\caption{Summary of prior and new upper and lower bounds for the PPP.}
\label{tab:ppp-bounds}
\end{table}

In Section \ref{sec:deterministicalgorithms}, we prove that the deterministic competitive ratio is exactly $K$, closing a factor 24 gap between prior upper and lower bounds. Our algorithm is simple and serves as a warm-up to the randomized algorithm.

\begin{theorem} \label{thm:deterministic}
    The deterministic competitive ratio for the online PPP is exactly $K$.
\end{theorem}

Our randomized algorithm uses more sophisticated techniques. Buchbinder and Naor~\cite{coveringpackinglp} previously observed that Meyerson's $O(\log K)$-competitive randomized algorithm can be recovered via the primal-dual framework, since the \emph{Laminar} PPP is a covering problem with row-sparsity $K$. However, if we avoid the reduction to the laminar case and the corresponding constant factor loss, the central difficulty is that the linear program of the non-laminar version has unbounded row-sparsity, which means directly using generic online covering algorithms would result in unbounded competitive ratio for fixed $K$. Our main contribution is doing multiplicative updates over the \emph{sum} of relevant variables, rather than the individual variables alone.

\begin{theorem} \label{thm:randomized}
    The randomized competitive ratio for the online PPP is at most $\ln K + \ln \ln K + O(1)$ as $K \rightarrow \infty$.
\end{theorem}

We also provide an improved lower bound, showing our algorithm's competitive ratio only asymptotically differs from the lower bound by an \emph{additive} constant of 2. In contrast, prior bounds left a \emph{multiplicative} gap of $32\ln 2\approx 22.18$.

\begin{theorem}\label{thm:randLB}
    Any randomized online algorithm for the PPP with $K$ permits has a competitive ratio of at least $S_K = \ln K + \ln \ln K + o(1)$, defined by $S_1 := 1$ and $S_{K+1} := S_K + \frac{S_K}{\exp(S_K)-1}$.
\end{theorem}

We conjecture that this lower bound is exactly tight. At a high level, it is based on a recursive combination of ski rental lower bounds.

\subsection{Related Work}
\subsubsection{The Parking Permit Problem}
Several other online problems involving purchasing dilemmas preceded the PPP. The classic online ski rental problem is one of the most well-understood problems in online algorithms, encapsulating the dilemma of whether to rent or buy a particular resource. For this problem, the competitive ratio is 2 for deterministic algorithms, and $e/(e-1)$ for randomized algorithms~\cite{KarlinMMO94}. It corresponds to the special case of the PPP where $K=2$ and the more expensive permit has infinite duration.
%The Multi-option Ski Rental is a direct precursor to the PPP, where there are multiple ski types, each lasting a fixed number of ski days and with distinct costs. The multi-option ski rental problem is no harder than the original problem, in the sense that the same competitive ratios are identical.

The PPP was first introduced by Meyerson \cite{ppporiginal}. Whereas many other problems considered purchases of resources which then last forever, the PPP is one of the earliest problems using a \emph{leasing optimization model}, where resources have a fixed lifespan. Meyerson presented upper and lower bounds for deterministic and randomized algorithms, which are matching up to multiplicative constants. In particular, the paper showed the competitive ratio for deterministic algorithms was between $\frac{K}{6}$ and $4K$, and the competitive ratio for randomized algorithms is between $\log_2(K)/2$ and $16 \ln K$. This work also introduced the laminar version of the problem, on which their algorithms and lower-bound results were based. Additionally, they showed applications to variants of Steiner forest problems.

A generalization of PPP, called the 2D PPP, was studied by Hu et al.~\cite{2dppp}. In this problem, demands vary along two axes, the time and also over required capacities. This setting has applications in cloud computing, where for example the number of virtual machines required may vary over time based on the popularity of a website. This paper demonstrated that the competitive ratio of this problem for deterministic algorithms is $\Theta(K)$, which was achieved using pseudo-polynomial time algorithms. De Lima et al. \cite{multippp} proposed another generalization of PPP, the Multi-PPP, in which the online demands can be any non-negative integer. They demonstrated a competitive ratio preserving pseudo-polynomial time reduction from the Laminar Multi-PPP to Laminar PPP, which involves running an arbitrary number of PPP instances concurrently. Therefore, Meyerson's previous results (which were also based on Laminar PPP) also hold for Multi-PPP. They also demonstrated that this could be turned into a strictly polynomial-time algorithm, while losing a factor of 2 in the competitiveness. They showed this technique can also convert Hu et al.'s 2D PPP algorithm into a strictly polynomial-time algorithm.

Since the introduction of the PPP, many other problems have been reformulated in the leasing optimization model, including Online Set Cover with Leasing \cite{setcoverleasing} and Leasing Non-Metric Facility Location~\cite{vertexcoverfacilitylocationleasing}.

There has also been some work studying the PPP in learning-augmented settings. The first, by Kharchenko and Kononov \cite{threeppp}, focused on the case of $K=3$. A recent preprint by Ameli, Sanita and Venzin \cite{coveringproblemsaugmented} proposed a general framework for learning-augmented algorithms for covering/packing problems, with the PPP as a special case, along with Set Cover, Non-Metric Facility Location, and Metric Facility Location. Coester, Tudose and Turoczy~\cite{CoesterTT26} present learning-augmented algorithms based on dual predictions for a number of online minimization problems, including the parking permit problem.

\subsubsection{Primal-Dual Algorithms}
For decades, LP duality theory has been utilized for the design and analysis of exact algorithms \cite{maximumflow, edmondsflowers, improvedflow, edmondsmatching, submodular} and approximation algorithms \cite{setcover, lpintegralitygap, constrainedforestproblems, steinertrees}. In the field of online algorithms, Buchbinder and Naor first developed a primal-dual framework for covering-packing problems across several papers \cite{coveringpackinglp, onlinerouting, pdsurvey}. %In particular, for online set cover in a universe with $n$ elements and $m$ sets, they developed a randomized $O(\log m\log n)$-competitive algorithm, with asymptotically matching lower bounds. 
In light of their work, it has been observed that several other earlier papers can be re-interpreted from a primal-dual perspective \cite{onlinebipartitite, onlinepaging}, leading to fruitful extensions \cite{onlinebipartititeextended, pdsurvey}, including convex objectives~\cite{AzarBCCCG0KNNP16}, and new breakthroughs on problems such as weighted paging~\cite{BansalBN12}.

\subsection{Preliminaries}

We write $\alg(I)$ for the cost of an algorithm $\alg$ on instance $I$, and similarly $\opt(I)$ for the optimal offline algorithm. The algorithm is $\rho$-competitive if $\alg(I)\le \rho\cdot\opt(I)$ for any instance $I$, and the minimal such $\rho$ is also called the competitive ratio. We will often drop $I$ from the notation when it is clear from the context and write $\alg$ and $\opt$ for both the names of the algorithms and their cost.

For $k=1,\dots,K$, we denote by $D_k\in\mathbb Z_{>0}$ the duration and by $C_k\in\mathbb R_{>0}$ the cost of permit type $k$. A sequence of days $\boldsymbol \sigma=(\sigma_1,\dots,\sigma_T)$ is revealed online, where $\sigma_t=1$ denotes that day $t$ is rainy and $\sigma_t=0$ otherwise. %We now introduce the linear programming formulation of the PPP.

By letting $x_{k, t}$ indicate the algorithm's decision to buy a permit of type $k$ on day $t$, we introduce a formulation of the problem as an integer linear program as follows.

\begin{equation} \label{program:pppPrimal}
        \begin{array}{ll@{}lll}
        \text{Minimize}  & \displaystyle\sum\limits_{t=1}^{T} \sum\limits_{k=1}^{K} C_kx_{k, t} &\\
        \text{subject to}& \displaystyle\sum\limits_{k=1}^K \sum\limits_{s=t-D_k+1}^t  &x_{k, s} \geq \sigma_t,  &t=1 ,\dots, T\\
                         &                                                &x_{k, t} \in \mathbb{Z}_{\geq 0},  &t = 1, \dots, T, k=1, \dots, K
        \end{array}
    \end{equation}

By relaxing the integrality constraint of \eqref{program:pppPrimal}, we arrive at a primal covering linear program and its corresponding dual packing program. We will call the formulation of \eqref{program:pppprimalfractional} the \emph{Fractional PPP}.

\begin{empheq}[box=\fbox]{equation} \label{program:pppprimalfractional}
\begin{aligned}
\textbf{Primal:} \\
\text{Minimize} \quad & \sum_{t=1}^{T} \sum_{k=1}^{K} C_k x_{k,t} \\
\text{subject to} \quad & \sum_{k=1}^K \sum_{s=t-D_k+1}^t x_{k,s} \ge \sigma_t, 
    && t=1,\dots,T \\
& x_{k,t} \ge 0, 
    && t=1,\dots,T, \; k=1,\dots,K
\end{aligned}
\end{empheq}

\begin{empheq}[box=\fbox]{equation} \label{program:dual}
\begin{aligned}
\textbf{Dual:} \\
\text{Maximize} \quad & \sum_{t=1}^{T} \sigma_t y_t \\
\text{subject to} \quad & \sum_{s=t}^{t+D_k-1} y_s \le C_k, 
    && t=1,\dots,T, \; k=1,\dots,K \\
& y_t \ge 0, 
    && t=1,\dots,T
\end{aligned}
\end{empheq}

One way to interpret the dual program is the following. Each variable $y_t$ can be interpreted as representing the amount of blame we assign to day $t$ for the cost of the PPP instance. The dual constraint $\sum_{s=t}^{t+D_k-1} y_s \le C_k$ is then interpreted as ensuring that the assignment of blame is consistent; we cannot blame more than $C_k$ to a contiguous sequence of $D_k$ days, as we could cover this sequence with a single permit costing $C_k$.

\section{Warm-Up: Deterministic Algorithms} 
\label{sec:deterministicalgorithms}

We begin by presenting our deterministic online algorithm and lower bound, which serve as a warm-up to the randomized results in the subsequent section.

\subsection{Deterministic Primal-Dual Algorithm}

We provide a simple $K$-competitive primal-dual algorithm for the deterministic PPP, given in Algorithm \ref{alg:kcompdeterministic}. %This same algorithm is also $K$-competitive for Multi-PPP, which is proven in \ref{subsec:kcompetitivemultippp}.

\begin{algorithm}[h]
    \caption{$K$-competitive deterministic PPP algorithm}\label{alg:kcompdeterministic}
    \begin{algorithmic}[1]

    %\Procedure{OnlineAlgorithm}{} \label{proc:stabledual}
    \State Initialize $\mathbf x \gets \mathbf 0$ and $\mathbf y \gets \mathbf 0$.
    \For{each rainy day $t$ not covered by any permit}
        \State Increase $y_t$ continuously until $\sum_{s=t-D_k+1}^t y_s = C_k$ for some $k$.
        \State Purchase such permit $k$, i.e., set $x_{k,t}=1$.
    \EndFor
    %\EndProcedure
    \end{algorithmic}
    \end{algorithm}
    \begin{lemma} \label{lem:kcompetitivealg}
        Algorithm \ref{alg:kcompdeterministic} is a $K$-competitive algorithm for the PPP.
    \end{lemma}
    \begin{proof}
        Let $P$ and $D$ represent the primal and dual values, respectively. We prove the following claims.
        \begin{enumerate}
        \item \label{claim:feasibleprimal} Algorithm \ref{alg:kcompdeterministic} produces a feasible online primal solution $\mathbf x$.
        \item At all points in the algorithm $P \leq K \cdot D$ \label{claim:primalbounded}
        \item Algorithm \ref{alg:kcompdeterministic} produces a feasible dual solution $\mathbf y$. \label{claim:feasibleonlinedual}
    \end{enumerate}

    Claim \ref{claim:feasibleprimal} is clear, as all rainy days will eventually be covered by a permit, because by increasing $y_t$, eventually $\sum_{s=t-D_k+1}^t y_s$ will reach some $C_k$.

    Claim \ref{claim:primalbounded} holds because the maximum amount spent on a fixed permit type $k$ is at most $D$. To see this, observe that we only purchase type $k$ at time $t$ if $\sum_{s=t-D_k+1}^t y_s = C_k$, but then we do not buy any new permits for $D_k$ days. Therefore, each purchase of cost $C_k$ for a type $k$ permit can be uniquely associated with a dual cost of $C_k$. Any time we make a purchase of a type $k$ permit costing $C_k$, we can charge this to the dual cost of $C_k$, and each charge must be unique as each type $k$ permit purchase is separated from the next by a time duration of $D_k$. This proves the claim.

    Claim \ref{claim:feasibleonlinedual} is evident, as no dual constraint can be violated before buying a permit.%, by line \ref{line:ifstatement}.

    With these claims proven, we appeal to weak duality to obtain that
    \begin{align*}
        \alg = \sum_{s=1}^T\sum_{k=1}^K x_{k, t} \le K \cdot \sum_{s=1}^T y_s \leq K\cdot \opt.&\qedhere
    \end{align*}
    \end{proof}

\subsection{Deterministic Lower Bound}
In this section, we prove a lower bound on the deterministic competitive ratio of the PPP, which matches our upper bound of $K$. 

\begin{lemma} \label{lem:detalglowerbound}
    The deterministic competitive ratio for the PPP is at least $K$
\end{lemma}
\begin{proof}
We choose $N$ to be some large integer constant. For $1 \leq k \leq K$, we take $C_k = N^{k-1}$ and $D_k = N^{2({k-1})}$. The adversarial sequence is such that a day is rainy exactly when $\alg$ has no valid permit at the start of said day. 
The input of rainy days terminates once a type $K$ permit has been purchased or once $D_K$ days have elapsed, whichever occurs first. 
We assume without loss of generality that $\alg$ only possesses at most one permit at a time. 
Let us define the problem instance with $K$ permits by $I_K$.

We will show by induction on $K$ that for any online algorithm $\alg$, we must have $\alg(I_K) \ge (K-\epsilon_{K,N})\opt(I_K)$, where $\epsilon_{K,N}\to 0$ as $N\to\infty$.  For $K=1$, this is trivial because the online algorithm cannot perform better than the offline optimal algorithm. We move on to the inductive step, showing that if the result holds up to $K$, then it holds for $K+1$ too. We proceed via a case analysis.

\paragraph*{Case 1: $\boldsymbol{\alg}$ does not buy a type $\boldsymbol{K+1}$ permit during $\boldsymbol{I_{K+1}}$}
In this case, $\alg$ will only use permits of type $K$ or less. Observe that due to the adversarial input, the algorithm must have a permit for every day during $I_{K+1}$. The cheapest option per day is the type $K$ permit, as $\frac{C_{k}}{D_{k}} = \frac{N^{k-1}}{N^{2(k-1)}} = \frac{1}{N^{k-1}}$ is strictly decreasing with $k$. As there are $D_{K+1} = N^{2K}$ days, this means that $\alg \geq \frac{1}{N^{K-1}} \cdot N^{2K} = N^{K+1}$. Observe however that $\opt \leq C_{K+1} = N^{K}$, as one strategy for this interval is just to buy the permit of type $K+1$. Hence we have $\alg \geq N^{K+1} \geq N \cdot \opt$. For $N$ large enough, this satisfies the required bound.

\paragraph*{Case 2: $\boldsymbol{\alg}$ does buy a type $\boldsymbol{K+1}$ permit during $\boldsymbol{I_{K+1}}$}

Note that $I_{K+1}$ is just a concatenation of several (complete) $I_K$ instances and a single (possibly incomplete) instance at the end during which $\alg$ buys the type $K+1$ permit. On the prior complete instances, the algorithm's cost is at most $\alg-C_{K+1}$ and the optimal cost is at least $\opt - C_K$, because the final incomplete instance contributes at most $C_K$ to $\opt$'s cost as it could cover this instance with a single permit of type $C_K$. Thus, applying the induction hypothesis to all complete instances yields
\begin{align*}
    \alg - C_{K+1} \ge (K-\epsilon_{K,N})\cdot(\opt-C_K).
\end{align*}

% Suppose that $\alg$ buys the permit of type $K$ on day $t^*+1$. Let $\alg_{\leq t^*}$ denote the cost incurred by the algorithm for times up to $t^*$. Let  $\opt_{\leq t^*}$ denote the cost of an optimal algorithm for an instance which \emph{only} consists of the first $t^*$ days. %By the inductive hypothesis, $\frac{\alg_{\leq t^*}}{\opt_{\leq t^*}} \geq K$, because the time prior to $t^*$ is just repeated $I_K$ instances. 
% By the inductive hypothesis, $\alg_{\leq t^*}\ge (K-\epsilon_{K,N})\cdot(\opt_{\leq t^*} - C_K)$, because the time prior to $t^*$ is just repeated $I_K$ instances and the lower bound of $K-\epsilon_{K,N}$ holds in each of them except possibly the last incomplete one, which contributes at most $C_K$ to the optimal cost. 
% Because $\alg$ buys the type $K+1$ permit at time $t^*+1$, and $\alg_{\leq t^*}$ accounts for all $\alg$'s costs prior to $t^*$ in $I_{K+1}$, we have $\alg \geq \alg_{\leq t^*} + C_{K+1}$. On the other hand, we have $\opt \leq \opt_{\leq t^*} + C_1$, as one strategy is for $\opt$ to use the strategy of $\opt_{\leq t^*}$, and then simply buy the type $1$ permit on day $t^*+1$. Combining these inequalities, we get
% \begin{align*}
%     \alg \ge (K-\epsilon_{K,N})\cdot(\opt - C_1 - C_K) + C_{K+1}.
% \end{align*}

Using $\opt \leq C_{K+1}$ (because one solution to $I_{K+1}$ is to simply buy the type $K+1$ permit immediately) and rearranging, we get
\begin{align*}
    \alg + KC_K\ge (K-\epsilon_{K,N}+1)\cdot\opt.
\end{align*}
Because $\alg \ge C_{K+1} = N\cdot C_K$, the left hand side is at most $\alg\cdot(1+\frac{K}{N})$. Thus,
\begin{align*}
    \frac{\alg}{\opt} \ge \frac{K-\epsilon_{K,N} + 1}{1+\frac{K}{N}}\to K+1\qquad\text{ as $N\to\infty$,}
\end{align*}
which concludes the induction step.
\end{proof}
% Multiplying by $\frac{1}{1+o(1)} = 1-o(1)$, the induction step holds with $\epsilon_{K+1,N} = \epsilon_{K,N}+o(K)$.

% We split Case 2 into two further cases.

% \paragraph*{Case 2i: $\boldsymbol{\mathrm{opt}_{\leq t^*} \leq \frac{C_{K+1}}{K+1} = \frac{N^K}{K+1}}$}

% Intuitively, this case corresponds to $\alg$ buying the permit of type $K+1$ too early. This yields
% \begin{equation}
%     \begin{aligned}
%         \frac{\alg}{\opt} &\geq \frac{\alg_{\leq t^*} + C_{K+1}}{\opt_{\leq t^*} + C_1}\\
%         &\geq \frac{C_{K+1}}{\frac{C_{K+1}}{K+1} + C_1}\\
%         &= \frac{N^K}{\frac{N^{K}}{K+1} + 1}\\
%         &= \frac{K+1}{1 + \frac{K+1}{N^K}}\\
%         &\rightarrow K+1 \text{ as } N \rightarrow \infty
%     \end{aligned}
% \end{equation}

% This thereby upholds the claim for $K+1$.

% \subsubsection*{Case 2ii: $\opt_{\leq t^*} \geq \frac{C_K}{K+1}= \frac{N^{K}}{K+1}$}

% This yields the following set of inequalities

% \begin{equation}
%     \begin{aligned}
%         \frac{\alg}{\opt} &\geq \frac{\alg_{\leq t^*} + C_{K+1}}{\opt}\\
%         &= \frac{\alg_{\leq t^*}}{\opt} + \frac{C_{K+1}}{\opt}\\
%         &\geq \frac{\alg_{\leq t^*}}{\opt_{\leq t^*}+C_1} + \frac{C_{K+1}}{C_{K+1}}\\
%         &= \frac{\alg_{\leq t^*}/\opt_{\leq t^*}}{1+C_1/\opt_{\leq t^*}} + 1\\
%         &\geq \frac{\alg_{\leq t^*}/\opt_{\leq t^*}}{1+\frac{K+1}{N^K}} + 1\\
%         &\rightarrow K+1 \text{ as } N \rightarrow \infty
%     \end{aligned}
% \end{equation}

% thereby finishing the case analysis, and proving that the competitive ratio is at least $K$.

\section{Randomized Algorithms}

In this section, we develop a randomized online algorithm using the primal-dual framework. We also show an improved lower bound on the randomized competitive algorithm, which is within an \emph{additive} constant of our algorithm's competitive ratio. %In particular, we prove the following

\subsection{Randomized Primal-Dual Algorithm} \label{sec:randomizedonlinealgorithm}

In this section, we develop a near-optimal randomized online algorithm, satisfying the guarantee of the following lemma:

\begin{lemma} \label{lem:randomizedupperbound}
    The randomized competitive ratio for the PPP is at most $\ln K + \ln \ln K + O(1)$.
\end{lemma}

In Meyerson's \cite{ppporiginal} randomized online algorithm, the most significant deterioration to the competitive ratio occurs due to the assumption of laminarity. In this section, we develop novel algorithmic techniques which show how the non-laminar structure can be handled directly, without factor-losing reductions to the laminar case. 

%We introduce improvements to the randomized rounding of the fractional solutions to the PPP, as well as an improved online algorithm for the fractional PPP itself.

Firstly, we present Algorithm \ref{alg:onlinefractional}, a primal-dual online algorithm to the \emph{fractional} PPP, which we show in Lemma \ref{lem:deterministicfractionalalgorithm} achieves a desirable competitive ratio.  Afterwards, we show that this can be rounded into an integral solution online, without degrading the competitive ratio. 

Algorithm \ref{alg:onlinefractional} draws inspiration from previous primal-dual algorithms \cite{pdsurvey}, such as those for set cover. In many primal-dual algorithms, each primal variable increases as an exponential function of some corresponding dual variables. For problems like set cover, this results in a competitive ratio logarithmic in the maximum number of sets containing any element (corresponding to the row sparsity in general covering problems). However, when considering non-laminar PPP as an instance of set cover, even for fixed $K$ this number is unbounded. Instead, on each day $t$, we orchestrate the increase to variables $x_{k,t}$ instead so that the \emph{summation} $\sum_{t-D_k+1}^t x_{k,t}$ increases exponentially with the dual variables. As shall be shown, this results in a competitive ratio being a function of $K$.

\begin{algorithm}[h]
    \caption{Online Fractional Algorithm for the fractional PPP}\label{alg:onlinefractional}
    \begin{algorithmic}[1]

    \State Initialize all $x_{k, t}$ and $y_t$ to $0$, for $1 \leq k \leq K, 1 \leq t \leq T$.
    \For{each rainy day $t$}
    \While{$\sum_{k=1}^K \sum_{s=t-D_k+1}^t x_{k, s} < 1$} 
        \State Increase $y_t$ uniformly.
        \State For all $k$, $x_{k, t} := \left(\frac{1}{K \ln K} + \sum_{s=t-D_k+1}^{t-1} x_{k,s}\right) \left(\exp\left(\frac{y_t}{C_k}\right) - 1 \right)$ \label{line:updaterule}
    \EndWhile
    \EndFor
    \end{algorithmic}
    \end{algorithm}

    \begin{lemma} \label{lem:deterministicfractionalalgorithm}
        Algorithm \ref{alg:onlinefractional} is $(\ln K+\ln \ln K + 2 + o(1))$-competitive for fractional PPP.
    \end{lemma}
    \begin{proof}

    Let $P$ and $D$ represent the primal and dual objective values throughout the algorithm. We prove this result in 3 steps.

    \begin{enumerate}
        \item \label{step:feasible} Algorithm \ref{alg:onlinefractional} produces a feasible online primal solution $\mathbf x$.
        \item \label{step:disparity} At all times in the algorithm $P \leq (1+\frac{1}{\ln K}) D$
        \item \label{step:dualfeasibility} Algorithm \ref{alg:onlinefractional} produces a dual solution $\mathbf y$ which violates each constraint by a factor no greater than $\ln(1+K \ln K)$.
    \end{enumerate}

    For Claim \ref{step:feasible}, observe that throughout the while loop, $x_{k,t}$ strictly increases, so eventually $\sum_{k=1}^K \sum_{s=t-D_k+1}^t x_{k, s} = 1$ is satisfied. Additionally, non-negativity is satisfied because at $y_t = 0$, $x_{k,t}=0$, so $x_{k,t}=0$ increases continuously and is always non-negative.
    
    Moving to Claim \ref{step:disparity}, observe that this is clearly satisfied initially, as $P=D=0$. We are left with showing that for a rainy day $t$, $\frac{\partial P}{\partial y_t} \leq (1 + \frac{1}{\ln K}) \frac{\partial D}{\partial y_t}$. Observe

    \begin{equation*}
        \begin{aligned}
            \frac{\partial P}{\partial y_t} &= \sum_{k=1}^K C_k \frac{\partial x_{k,t}}{\partial y_t}\\
            &= \sum_{k=1}^K C_k \left(\frac{1}{K \ln K} + \sum_{s=t-D_k+1}^{t-1} x_{k,s}\right) \cdot \frac{1}{C_k}\cdot\exp\left(\frac{y_t}{C_k} \right)\\
            &= \sum_{k=1}^K \left(\frac{1}{K \ln K} + \sum_{s=t-D_k+1}^{t-1} x_{k,s}\right) \cdot \left(\exp\left(\frac{y_t}{C_k} \right)-1+1\right)\\
            &= \sum_{k=1}^K \left( x_{k,t} + \frac{1}{K \ln K} + \sum_{s=t-D_k+1}^{t-1} x_{k,s}\right) \\
            &= \sum_{k=1}^K \left(\frac{1}{K \ln K} + \sum_{s=t-D_k+1}^{t} x_{k,s}\right) \\
            &\leq \frac{1}{\ln K}+1 \\
            &= \left(1 + \frac{1}{\ln K} \right) \cdot \frac{\partial D}{\partial y_t}
        \end{aligned}
    \end{equation*}

    We finally turn to Claim \ref{step:dualfeasibility}. We will show that $\mathbf y$ cannot violate any $(k,t)$-constraint by more than a factor of $\ln(1 +K \ln K)$. Fix $(k,t)$. We seek to show that $\sum_{s=t-D_k+1}^t y_{s} \leq C_k \cdot \ln(1 +K \ln K)$. We claim for all $t'$ with $t-D_k+1 \leq t' \leq t$ that

    \begin{equation} \label{claim:upperboundsum}
        \sum_{s=t-D_k+1}^{t'} x_{k, s} \geq \frac{1}{K \ln K} \left( \exp \left( \frac{\sum_{s=t-D_k+1}^{t'} y_s}{C_k} \right) - 1\right)
    \end{equation}

    We prove this by induction on $t'$. The base case of $t' = t-D_k+1$ is immediate from the definition of $x_{k,t-D_k+1}$ and nonnegativity of $\mathbf x$.
    % evident, as

    % \begin{equation}
    %     \begin{aligned}
    %         \sum_{s=t-D_k+1}^{t'} x_{k, s} &= x_{k, t-D_k+1}\\
    %         &= \left(\frac{1}{K \ln K} + b\right) \left(\exp\left(\frac{y_{t-D_k+1}}{C_k} \right) - 1 \right)\\
    %         &\geq \frac{1}{K \ln K} \left(\exp\left(\frac{y_{t-D_k+1}}{C_k} \right) - 1 \right)\\
    %         &= \frac{1}{K \ln K} \left( \exp \left( \frac{\sum_{s=t-D_k+1}^{t'} y_s}{C_k} \right) - 1\right)
    %     \end{aligned}
    % \end{equation}

    For $t-D_k+1 < t' \leq t$, we use the following inductive argument:

    \begin{equation*}
        \begin{aligned}
            \sum_{s=t-D_k+1}^{t'} x_{k, s} &=  x_{k, t'} + \sum_{s=t-D_k+1}^{t'-1} x_{k, s}\\
            &\geq \left(\frac{1}{K \ln K} + \sum_{s=t-D_k+1}^{t'-1} x_{k,s}\right) \left(\exp\left(\frac{y_{t'}}{C_k} \right) - 1 \right) + \sum_{s=t-D_k+1}^{t'-1} x_{k, s}\\
            &= \frac{1}{K \ln K} \left(\exp\left(\frac{y_{t'}}{C_k} \right) - 1 \right) + \sum_{s=t-D_k+1}^{t'-1} x_{k,s} \exp\left(\frac{y_{t'}}{C_k}  \right)\\
            &\geq \frac{1}{K \ln K} \left(\exp\left(\frac{y_{t'}}{C_k} \right) - 1 \right) \\&+ \frac{1}{K \ln K} \left( \exp \left( \frac{\sum_{s=t-D_k+1}^{t'-1} y_s}{C_k} \right) - 1\right) \exp\left(\frac{y_{t'}}{C_k}  \right)\\
            &= \frac{1}{K \ln K} \left(\exp \left( \frac{\sum_{s=t-D_k+1}^{t'} y_s}{C_k} \right) - 1\right)
        \end{aligned}
    \end{equation*}
    as required. Specializing \eqref{claim:upperboundsum} for $t' = t$, and observing by the algorithm that $\sum_{s=t-D_k+1}^{t} x_{k, s} \leq 1$, we find that $\frac{1}{K \ln K} \left( \exp \left( \frac{\sum_{s=t-D_k+1}^{t} y_s}{C_k} \right) - 1\right) \leq 1$, which can be rearranged into $\sum_{s=t-D_k+1}^t y_{t} \leq C_k \cdot \ln(1+K\ln K)$, as desired.

    From the above claims, the competitive bound follows quickly. From Claim \ref{step:dualfeasibility}, we know that $\hat{\mathbf y}:= \frac{\mathbf y}{\ln(1+K\ln K)}$ is a feasible dual solution. Weak duality implies that the objective value of $\hat{\mathbf y}$ is a lower bound on the optimal primal solution value. Hence, $\frac{1}{\ln(1+K\ln K)} \sum_{t=1}^T y_t \leq \opt$. By Claim~\ref{step:disparity}, we hence know that
    \begin{equation*}
        \alg \leq \left(1+\frac{1}{\ln K}\right) \sum_{t=1}^T y_t \leq \left(1+\frac{1}{\ln K}\right)(1+\ln(1+K\ln K)) \opt
    \end{equation*}

    Then, observing that $\ln(1 + K \ln K) = \ln K + \ln \ln K + o(1)$ we have
    \begin{equation*}
        \begin{aligned}
                \left(1+\frac{1}{\ln K}\right)(1+\ln(1+K\ln K)) &= \left(1+\frac{1}{\ln K}\right)(1+\ln K + \ln \ln K + o(1))\\
                &= 2 + \ln K + \ln \ln K + o(1),
        \end{aligned}
    \end{equation*}
    concluding the proof.
    \end{proof}

Next, we show that there exists an online rounding algorithm for the fractional PPP, which creates a randomized integral solution whose expected cost is no greater than the fractional solution. Paired with Lemma \ref{lem:deterministicfractionalalgorithm}, this immediately proves Lemma \ref{lem:randomizedupperbound}. Meyerson \cite{ppporiginal} proved a similar result, however the reduction was not perfect, as the conversion from a deterministic fractional PPP to randomized PPP lost a factor of 8 in the competitive ratio. 

    \begin{lemma} \label{lem:randomization}
        There exists an online randomized rounding algorithm, which takes as input an online instance $\boldsymbol \sigma$ and an online algorithm for the fractional PPP with cost $\alg_{\textup{frac}}$, and constructs a randomized online integral solution to the PPP with cost $\alg$ such that $\mathbb{E}[\alg] \leq \alg_{\textup{frac}}$.
    \end{lemma}

    \begin{proof}
    Fix an online instance $\boldsymbol \sigma$, and suppose we have an online algorithm $\alg_{\textup{frac}}$ for the fractional PPP. We assume that \textbf{(\refstepcounter{assump}\label{assump:rainyinc}\theassump)}  on day $t$, $\alg_{\textup{frac}}$ may only increase the variables $x_{k,t}$, and that \textbf{(\refstepcounter{assump}\label{assump:massone}\theassump)} rainy days are covered by a total mass of 1 permit. These assumptions are without loss of generality, as any $\alg_{\textup{frac}}$ can be converted into an algorithm satisfying this without increasing the objective value. %, in polynomial time. 
    We claim that Algorithm \ref{alg:reduction} is then a randomized integral solution with $\mathbb E[\alg] \leq \alg_{\textup{frac}}$.

\begin{algorithm}[t] % 't' is usually preferred over 'h!' in LIPIcs
    \caption{Deterministic fractional PPP to randomized integral PPP conversion}\label{alg:reduction}
    \begin{algorithmic}[1]
        \Require Fractional online algorithm $\alg_{\textup{frac}}$ satisfying assumptions \eqref{assump:rainyinc} and \eqref{assump:massone}.
        
        \State Initiate online fractional solution $\mathbf x \gets \mathbf 0$
        
        \For{each rainy day $t$ not covered by any permit} \label{line:rainypermitlessdays}
            \State Use $\alg_{\textup{frac}}$ to compute fractional online solution $\mathbf x$ up to day $t$.
            \State Purchase exactly one of the permits at random, where permit $k$ is chosen with probability $\frac{x_{k,t}}{\sum_{k'} x_{k', t}}$.
        \EndFor
    \end{algorithmic}
\end{algorithm}

Firstly, observe that the algorithm always produces a feasible integral solution, as every day which begins without a permit will have a new permit purchased. Next, let $Z_{k,t}$ be a random indicator variable which indicates whether the permit of type k is purchased at time $t$ by the algorithm. Then, we have that $\alg = \sum_{k} C_k \sum_{t=1}^T Z_{k,t}$. To prove that $\mathbb E [\alg] \leq \alg_{\textup{frac}}$, it therefore suffices to show that $\mathbb E[Z_{k,t}] \leq x_{k,t}$ for all $k$ and $t$ by linearity of expectation. The inequality is obvious for days $t$ with $\sigma_t = 0$, because in that case both sides of the equation equal 0. Focusing instead on days with $\sigma_t = 1$, let $t_j$ denote the $j$th day with $\sigma_t = 1$. We prove by induction on $j$ that $\mathbb E[Z_{k, t_j}] = x_{k,t_j}$. %Let $X_{k} := \sum_{t=1}^T X_i^t$, which indicates whether permit $i$ is ever purchased. Then this will immediately also prove that $\mathbb E[X_{k}] \leq \sum_{s=t-D_k+1}^t x_{k, s}$ by disjoint additivity.

    For the base case, consider $j=1$. Observe that $\sum_{k} x_{k, t_1} = 1$ as $x_{k, t} = 0$ for times $t < t_1$, and day $t_1$ must be covered on day $t_1$ for the fractional solution to be feasible. Therefore, on day $t_1$, the algorithm purchases permit $k$ with probability $\frac{x_{k,t_1}}{\sum_{k'}x_{k', t_1}} = x_{k,t_1}$. This means $\mathbb E[Z_{k, t_1}] = x_{k, t_1}$, thereby demonstrating the base case.

    For the inductive step, fix $k$ and $j$, and assume the claim holds for all $j' < j$ and for all $k$. In order for $t_j$ to be a rainy day where $Z_{k,t_j}$ even has the possibility of being non-zero, it must be that $\alg$ purchased a permit which became invalid in $(t_{j-1}, t_j]$. Let $I_j$ be the event that $\alg$ owned a permit which became invalid at time $(t_{j-1}, t_j]$. The necessity of $I_j$ for $Z_{k, t_j} = 1$ implies 
    \begin{equation}
    \mathbb{E}[Z_{k, t_j}] = \mathbb{P}[Z_{k, t_j} = 1] = \mathbb{P}[Z_{k, t_j} = 1 \cap I_j] = \mathbb{P}[Z_{k, t_j} = 1 | I_j] \cdot \mathbb P[I_j]. \end{equation}

    Now, $\mathbb{P}[Z_{k, t_j} = 1 | I_j] = \frac{x_{k, t_j}}{\sum_{k'}x_{k', t_j}}$ by definition of Algorithm \ref{alg:reduction}. On the other hand, by additivity of disjoint events and the inductive hypothesis, we get
    \begin{equation}
        \mathbb P[I_j] = \sum_{k=1}^K \sum_{t_{j-1} < s \leq t_j} \mathbb E[Z_{k, s-D_{k}}] = \sum_{k=1}^K \sum_{t_{j-1} < s \leq t_j} x_{k, s-D_{k}}. 
    \end{equation}
    However, we have that $\sum_{k=1}^K \sum_{t_{j-1} < s \leq t_j} x_{k, s-D_{k}} = \sum_{k'}x_{k',t_j}$, because the mass of permits which become invalid in $(t_{j-1}, t_j]$ must be made up by new mass to permits at time $t_j$. Therefore, we have $ \mathbb{E}[Z_{k, t_j}] = \frac{x_{k, t_j}}{\sum_{k'}x_{k', t_j}} \cdot \sum_{k'}x_{k', t_j} = x_{k, t_j}$, which finishes the inductive step and hence the proof.
    \end{proof}

\subsection{Randomized Lower bound}

We now provide a lower bound to the fractional PPP, matching the upper bound guarantee of Algorithm \ref{alg:onlinefractional} up to an additive constant. Since any randomized algorithm can be turned into a fractional algorithm with the same competitive ratio, by defining $x_{k,t}$ as the probability of purchasing permit $k$ at time $t$, this implies the lower bound of Theorem~\ref{thm:randLB} for randomized algorithms, 

We begin by showing that it suffices to lower bound the laminar version of the problem. Recall that in Laminar PPP, we require that $D_k/D_{k-1}$ is an integer for all $k=2,\dots,K$, and any permit of type $k$ expires on the next day that is a multiple of $D_k$. The linear program for the laminar version is the same as~\eqref{program:pppprimalfractional} except that variables $x_{k,t}$ exist only for $t=1+(i-1)D_k$ for some integer $i$ (i.e., all other $x_{k,t}$ are set to $0$).

\begin{lemma} \label{lem:laminarsufficient}
   For fixed $K$, the competitive ratio for online fractional PPP is at least the competitive ratio for online fractional \emph{Laminar} PPP.
\end{lemma}
\begin{proof}
    Fix $K$. Suppose there is an $\alpha$-competitive algorithm $\alg'$ for the online fractional (non-laminar) PPP. We will show that there is an $\alpha$-competitive algorithm $\alg$ for Laminar PPP.

    Consider an instance for Laminar PPP with durations $D_1,\dots,D_K$, costs $C_1,\dots,C_K$ and an online sequence $\boldsymbol \sigma$. We may assume without loss of generality that $D_1 = 1$. This implies that each day $t$ can be uniquely specified by a tuple $(i_1, \dots, i_K)$ with $1 \leq i_k \leq D_{k+1} / D_k$ for $k < K$, such that day $t$ is covered by the $i_K$th type $K$ permit, and covered by the $i_{K-1}$th permit among the $D_{K}/D_{K-1}$ permits of type $K-1$ overlapping this type $K$ permit, and so on. Formally, this mapping from tuple $(i_1,\dots,i_K)$ to day $t$ is given by
    \begin{align*}
        f(i_1,\dots,i_K) = 1 + \sum_{k=1}^K (i_k - 1)D_k.
    \end{align*}

    Algorithm $\alg$ will work by simulating $\alg'$ on a modified non-laminar instance with durations $D_{k}'$, costs $C_k'$ and an online sequence $\boldsymbol \sigma'$. We will define $D'_k = 2^{k-1} D_k$, so that an interval of type $k$ can keep twice as many intervals of type $k-1$ as it originally could.
    %We define $\mathcal I'(i_1, ..., i_K)$ analogously to $\mathcal I(i_1, ..., i_K)$, but using $\mathcal I'$ to define the durations of permits. 
    The costs $C_k'$ are kept identical to $C_k$. An online sequence $\boldsymbol \sigma$ for the original laminar instance is modified into an online sequence $\boldsymbol \sigma'$ as follows: if $t=f(i_1, \dots, i_K)$ is a rainy day in $\boldsymbol \sigma$, then the day $t'$ defined by $t' = f'(2i_1, \dots, 2i_K)$ is rainy in $\boldsymbol \sigma'$, where $f'$ is defined like $f$ but with $D_k$ replaced by $D_k'$. Let $\opt$ and $\opt'$ be the optimal solution values for $\boldsymbol \sigma$ and $\boldsymbol \sigma'$, respectively. One may immediately observe that $\opt' \leq \opt$, as we can convert a laminar solution for $\boldsymbol \sigma$ into a laminar solution for $\boldsymbol \sigma'$ without increasing the cost, and the optimal non-laminar solution for $\boldsymbol \sigma'$ is at least this good.
   
    We now wish to show $\alg \leq \alg'$, where $\alg$ and $\alg'$ denote the costs of the online algorithms' solutions for $\boldsymbol \sigma$ and $\boldsymbol \sigma'$, respectively. Algorithm $\alg$ works by simulating $\alg'$ on $\boldsymbol \sigma'$, where on a rainy day $t$, $\alg$ increases the weight of the permit of each type $k$ covering $t$ so that the mass of the type $k$ permit covering $t$ equals the total mass of type $k$ permits covering $t'$ by $\alg'$. Observe that $\alg$ is clearly feasible.

    We charge the cost of purchases by $\alg$ to the corresponding purchases by $\alg'$ that triggered them. To show that the cost of $\alg$ is at most that of $\alg'$, it suffices to show that if two days $t_1 \ne t_2$ are covered by $\alg$ using separate permits of type $k$, then they charge to distinct purchases by $\alg'$. Indeed, if $t_1=f(i_1,\dots,i_K)$ and $t_2 = f(j_1,\dots,j_K)$, then the corresponding days $t_1' = f'(2i_1,\dots,2i_K)$ and $t_2' = f'(2j_1,\dots,2j_K)$ are a distance of at least $D_k'$ apart. Thus, $\alg'$ cannot cover them with a single permit of type $k$.
    
    % We may assume without loss of generality that $\alg'$ only ever makes (fractional) purchases at the moment a day $t'$ is revealed to be rainy, and only purchases those permits starting from $t'$. We will argue that the cost of $\alg$ on $\boldsymbol \sigma$ is at most the cost of $\alg'$ on $\boldsymbol \sigma'$. To see this, observe that any cost incurred by $\alg$ on $t$ due to a fractional purchase of permit $k$ can be associated with an equal cost incurred by $\alg'$ due to a fractional purchase of a type $k$ permit starting from $t'$. Using a charging argument, charge this cost of $\alg$ to $\alg'$. Observe that the charges against $\alg'$ will be unique: this is because whenever we have two days $t_1 = \mathcal I(i_1, ..., i_K)$ and $t_2 = \mathcal I(i_1, ..., i_K)$ covered by separate type $k$ permits in the laminar permit set, it must be that $t_1' = \mathcal I(2i_1, ..., 2i_K)$ and $t_2' = \mathcal I(2i_1, ..., 2i_K)$ are a distance of at least $D_k$ apart. Therefore it is impossible for a single fractional purchase by $\alg'$ to be associate with several fractional purchases by $\alg$.
        
    Hence, we have $\frac{\alg}{\opt} \leq \frac{\textsc{alg'}}{\textsc{opt'}}$ which proves the lemma.
\end{proof}

In Lemma \ref{lem:randomizedlowerbound}, we show that the randomized competitive ratio is lower bounded by $S_K$ defined by the recurrence $S_1 := 1$ and $S_{K+1} := S_K + \frac{S_K}{\exp(S_K)-1}$. We remark that $S_2 = \frac{e}{e-1}$ is the known optimal randomized competitive ratio for ski rental~\cite{KarlinMMO94}. In fact, our lower bound is obtained by combining ski rental lower bounds recursively. We prove in Lemma \ref{lem:recurrenceproperty} that $S_K$ asymptotically matches the upper bound of Lemma \ref{lem:deterministicfractionalalgorithm} up to additive constants. We leave its proof to Appendix \ref{app:omittedproofs}.

\begin{restatable}{lemma}{lemrecurrenceproperty} \label{lem:recurrenceproperty}
    Let $S_K$ be defined by $S_1 := 1$ and $S_{K+1} := S_K + \frac{S_K}{\exp(S_K)-1}$. Then $S_K = \ln K + \ln \ln K + o(1)$ as $K \rightarrow \infty$.
\end{restatable}

\begin{lemma}  \label{lem:randomizedlowerbound}
    Any online algorithm for the fractional online PPP with $K$ permits has a competitive ratio of at least $S_K$, where $S_1 := 1$ and $S_{K+1} := S_K + \frac{S_K}{\exp(S_K)-1}$.
\end{lemma}

\begin{proof}
    By Lemma \ref{lem:laminarsufficient} we may prove the lemma for the Laminar PPP. We proceed by induction on $K$. The base case $K=1$ is obvious. Now considering $K \geq 2$, assume that for $K-1$ it is known that we have an adversarial strategy so that the competitive ratio is at least $S_{K-1}$. Let $C_k = N^k$ and $D_k = M^k$, where $M$ and $N$ are large integers and $M > N$. In particular, any type $K$ interval consists of $M$ disjoint type $K-1$ intervals. We construct an instsance of length $T = D_K$, so that the entire problem instance consists of only one type $K$ interval.

    We define our adversarial instance for $K$ as follows. We iteratively apply the $S_{K-1}$ adversarial case (assumed to exist by the inductive hypothesis) to the $i$th type $K-1$ interval, for $i=1, 2, \dots$. Let $\alg_i$ and $\opt_i$ denote cost of the online algorithm $\alg$ on the $i$th type $K-1$ interval, and $\opt_i$ the optimal offline cost on that interval. 

    If after the $m$th type $K-1$ interval we have $\sum_{i=1}^m\alg_i \geq S_K \min \{ \sum_{i=1}^m\opt_i,  C_K\}$, then we stop the problem instance, and observe that the algorithm on this instance would have a competitive ratio of $S_K$, as $\opt \leq \min \{ \sum_{i=1}^m\opt_i,  C_K\}$ in that case. As this would finish the proof, instead suppose that we have
    \begin{equation}
        \sum_{i=1}^m\alg_i \leq S_K \min \left\{ \sum_{i=1}^m\opt_i,  C_K\right\}
    \end{equation}
    at all times. Let $t_i = i\cdot D_{K-1}$ denote the final time step in the $i$th type $K-1$ permit. Also let $x_K(t_i)$ denote the mass of permit $K$ which has been purchased at time $t_i$. Now, observe that the cost of the algorithm at time $t_m$ is bounded by

    \begin{equation} \label{discretebounds}
        C_K x_K(t_m) + \sum_{i=1}^m S_{K-1}  \cdot \opt_i (1-x_K(t_i)) \leq \sum_{i=1}^m\alg_i \leq S_K\min \left\{ \sum_{i=1}^m\opt_i,  C_K\right\}
    \end{equation}

    We will use inequality \eqref{discretebounds} to bound each $x_K(t_i)$ from above and thereby arrive at a lower bound on the cost of $\alg$. By making $M\gg N\gg 1$ arbitrarily large, meaning each type $K-1$ interval is a vanishingly narrow part of the type $K$ permit, this PPP adversarial example becomes arbitrarily close to a continuous problem. Moreover, the minimum on the right-hand side will be $C_K$, and $x_K(t_M)$ will tend to $1$ as otherwise the left-hand side would grow unbounded. By letting $y$ be the continuous analog of $\sum_{i=1}^m \opt_i$, the minimum possible cost of $\alg$ is thus lower bounded by the solution to the following optimization problem:
\begin{equation}
    \min_{x_K} \quad C_K + S_{K-1} \int_0^{C_K} (1 - x_K(y))\, dy\label{eq:contObjective}
\end{equation}
subject to
\begin{equation}
    C_K x_K(y) + S_{K-1} \int_0^y (1 - x_K(s))\, ds \leq S_K \cdot y \quad \forall\, y \in [0, C_K]\label{eq:randLBcontraint}
\end{equation}
and $x_K(0) = 0$, $x_K$ non-decreasing, $x_K(y) \in [0,1]$. The minimum in \eqref{eq:contObjective} is achieved by maximizing the $x_K(y)$, and subject to the constraint~\eqref{eq:randLBcontraint}, $x_K(y)$ is maximized by greedily maximizing $x_K(s)$ for all $s\in[0,y]$. Thus, the constraint~\eqref{eq:randLBcontraint} is tight in the optimal solution, we can solve for $x_{K}(y)$ by differentiating both sides of the constraint with respect to $y$:
\begin{equation*}
    C_K x_K'(y) + S_{K-1}(1 - x_K(y)) = S_K
\end{equation*}
Rearranging and using the recurrence definition of $S_K$:
\begin{equation*}
    C_K x_K'(y) = S_K - S_{K-1} + S_{K-1} x_K(y) = \frac{S_{K-1}}{\exp(S_{K-1})-1} + S_{K-1} x_K(y)
\end{equation*}
Given the initial condition $x_K(0)=0$, this differential equation is solved by
\begin{equation*}
    x_K(y) = \frac{1}{\exp(S_{K-1})-1}\left(\exp\left(\frac{S_{K-1} y}{C_K}\right) - 1\right).
\end{equation*}
%Now, the offline optimum over the full type $K$ interval is $\opt = C_K$, as $M$ is arbitrarily large. 
Combining the lower bound~\eqref{eq:contObjective} on the algorithm's cost with the fact that $\opt\le C_K$, the competitive ratio satisfies:
\begin{equation*}
    \frac{\alg}{\opt} \geq 1 + \frac{S_{K-1}}{C_K}\int_0^{C_K}(1-x_K(y))\,dy
\end{equation*}
We now evaluate the integral:
\begin{align*}
    \int_0^{C_K}(1 - x_K(y))\,dy 
    &= \int_0^{C_K} \left(1 - \frac{1}{\exp(S_{K-1})-1}\left(\exp\left(\frac{S_{K-1}y}{C_K}\right)-1\right)\right)dy \\
    &= C_K + \frac{1}{\exp(S_{K-1})-1}\left(C_K - \frac{C_K}{S_{K-1}}\left(\exp(S_{K-1})-1\right)\right) \\
    &= C_K + \frac{C_K}{\exp(S_{K-1})-1} - \frac{C_K}{S_{K-1}} \\
    &= C_K\left(1 + \frac{1}{\exp(S_{K-1})-1} - \frac{1}{S_{K-1}}\right)
\end{align*}
Substituting back:
\begin{align*}
    \frac{\alg}{\opt} &\geq 1 + S_{K-1}\left(1 + \frac{1}{\exp(S_{K-1})-1} - \frac{1}{S_{K-1}}\right) \\
    &= S_{K-1} + \frac{S_{K-1}}{\exp(S_{K-1})-1} \\
    &= S_K
\end{align*}
where the last equality is the recurrence defining $S_K$. This completes the induction.
\end{proof}

\bibliography{esa_reference}

@inproceedings{ppporiginal,
  author       = {Adam Meyerson},
  title        = {The Parking Permit Problem},
  booktitle    = {46th Annual {IEEE} Symposium on Foundations of Computer Science, {FOCS}},
  pages        = {274--284},
  publisher    = {{IEEE} Computer Society},
  year         = {2005},
  url          = {https://doi.org/10.1109/SFCS.2005.72},
  doi          = {10.1109/SFCS.2005.72}
}

@article{multippp,
title = {On Generalizations of the Parking Permit Problem and Network Leasing Problems},
journal = {Electronic Notes in Discrete Mathematics},
volume = {62},
pages = {225-230},
year = {2017},
note = {LAGOS'17 – IX Latin and American Algorithms, Graphs and Optimization},
doi = {https://doi.org/10.1016/j.endm.2017.10.039},
author = {M.S. {de Lima} and M.C. {San Felice} and O. Lee}
}

@INPROCEEDINGS {2dppp,
author = { Hu, Xinhui and Ludwig, Arne and Richa, Andrea and Schmid, Stefan },
booktitle = { 2015 IEEE 35th International Conference on Distributed Computing Systems (ICDCS) },
title = {{ Competitive Strategies for Online Cloud Resource Allocation with Discounts: The 2-Dimensional Parking Permit Problem}},
year = {2015},
pages = {93-102},
doi = {10.1109/ICDCS.2015.18}
}

@inproceedings{setcoverleasing,
author = {Abshoff, Sebastian and Markarian, Christine and Meyer auf der Heide, Friedhelm},
title = {Randomized Online Algorithms for Set Cover Leasing Problems},
year = {2014},
doi = {10.1007/978-3-319-12691-3_3},
booktitle = {Combinatorial Optimization and Applications: 8th International Conference, COCOA},
pages = {25–34}
}

@inproceedings{vertexcoverfacilitylocationleasing,
  author       = {Christine Markarian and
                  Friedhelm Meyer auf der Heide},
  title        = {Online Algorithms for Leasing Vertex Cover and Leasing Non-metric Facility Location},
  booktitle    = {Proceedings of the 8th International Conference on Operations Research and Enterprise Systems, {ICORES}},
  pages        = {315--321},
  year         = {2019},
  doi          = {10.5220/0007369503150321}
}

@article{KarlinMMO94,
  author       = {Anna R. Karlin and
                  Mark S. Manasse and
                  Lyle A. McGeoch and
                  Susan S. Owicki},
  title        = {Competitive Randomized Algorithms for Nonuniform Problems},
  journal      = {Algorithmica},
  volume       = {11},
  number       = {6},
  pages        = {542--571},
  year         = {1994},
  doi          = {10.1007/BF01189993},
  bibsource    = {dblp computer science bibliography, https://dblp.org}
}

@inproceedings{threeppp,
author = {Yaroslav, Kharchenko and Alexander, Kononov},
title = {A Learning-Augmented Algorithm for the Parking Permit Problem with Three Permit Types},
year = {2024},
url = {https://doi.org/10.1007/978-3-031-62792-7_8},
doi = {10.1007/978-3-031-62792-7_8},
booktitle = {Mathematical Optimization Theory and Operations Research: 23rd International Conference, MOTOR},
pages = {116–126}
}

@misc{coveringproblemsaugmented,
      title={Learning-Augmented Online Covering Problems}, 
      author={Afrouz Jabal Ameli and Laura Sanita and Moritz Venzin},
      year={2025},
      eprint={2507.06032},
      archivePrefix={arXiv}
}

@article{pdsurvey,
  author       = {Niv Buchbinder and
                  Joseph Naor},
  title        = {The Design of Competitive Online Algorithms via a Primal-Dual Approach},
  journal      = {Found. Trends Theor. Comput. Sci.},
  volume       = {3},
  number       = {2-3},
  pages        = {93--263},
  year         = {2009},
  url          = {https://doi.org/10.1561/0400000024},
  doi          = {10.1561/0400000024}
}

@article{maximumflow, 
title={Maximal Flow Through a Network}, 
volume={8}, DOI={10.4153/CJM-1956-045-5}, 
journal={Canadian Journal of Mathematics}, 
author={Ford, L. R. and Fulkerson, D. R.}, 
year={1956}, 
pages={399–404}}

@article{edmondsflowers, 
title={Paths, 
Trees, and Flowers}, 
volume={17}, 
DOI={10.4153/CJM-1965-045-4}, 
journal={Canadian Journal of Mathematics}, 
author={Edmonds, Jack}, 
year={1965}, 
pages={449–467}}

@article{improvedflow,
author = {Edmonds, Jack and Karp, Richard M.},
title = {Theoretical Improvements in Algorithmic Efficiency for Network Flow Problems},
year = {1972},
publisher = {Association for Computing Machinery},
volume = {19},
number = {2},
url = {https://doi.org/10.1145/321694.321699},
doi = {10.1145/321694.321699},
journal = {J. ACM},
pages = {248–264}
}

@article{edmondsmatching,
  title={Maximum matching and a polyhedron with 0,1-vertices},
  author={Edmonds, Jack},
  journal={Journal of Research of the National Bureau of Standards B},
  volume={69},
  number={125-130},
  pages={55--56},
  year={1965}
}

@article{submodular,
  title={On submodular function minimization},
  author={Cunningham, William H},
  journal={Combinatorica},
  volume={5},
  number={3},
  pages={185--192},
  year={1985},
  publisher={Springer}
}

@article{setcover,
author = {Hochbaum, Dorit S.},
title = {Approximation Algorithms for the Set Covering and Vertex Cover Problems},
journal = {SIAM Journal on Computing},
volume = {11},
number = {3},
pages = {555-556},
year = {1982}
}

@article{lpintegralitygap,
author = {Lov\'{a}sz, L.},
title = {On the ratio of optimal integral and fractional covers},
year = {1975},
issue_date = {January, 1975},
publisher = {Elsevier Science Publishers B. V.},
address = {NLD},
volume = {13},
number = {4},
issn = {0012-365X},
url = {https://doi.org/10.1016/0012-365X(75)90058-8},
doi = {10.1016/0012-365X(75)90058-8},
journal = {Discrete Math.},
month = jan,
pages = {383–390},
numpages = {8}
}

@article{constrainedforestproblems,
author = {Goemans, Michel X. and Williamson, David P.},
title = {A General Approximation Technique for Constrained Forest Problems},
journal = {SIAM Journal on Computing},
volume = {24},
number = {2},
pages = {296-317},
year = {1995},
doi = {10.1137/S0097539793242618},
}

@inproceedings{steinertrees,
author = {Agrawal, Ajit and Klein, Philip and Ravi, R.},
title = {When trees collide: an approximation algorithm for the generalized Steiner problem on networks},
year = {1991},
url = {https://doi.org/10.1145/103418.103437},
doi = {10.1145/103418.103437},
booktitle = {Proceedings of the Twenty-Third Annual ACM Symposium on Theory of Computing, {STOC}},
pages = {134–144}
}

@inproceedings{coveringpackinglp,
  author       = {Niv Buchbinder and
                  Joseph Naor},
  title        = {Online Primal-Dual Algorithms for Covering and Packing Problems},
  booktitle    = {Algorithms - {ESA} 2005, 13th Annual European Symposium},
  series       = {Lecture Notes in Computer Science},
  pages        = {689--701},
  publisher    = {Springer},
  year         = {2005},
  url          = {https://doi.org/10.1007/11561071\_61},
  doi          = {10.1007/11561071\_61}
}

@inproceedings{onlinerouting,
  author       = {Niv Buchbinder and
                  Joseph Naor},
  title        = {Improved Bounds for Online Routing and Packing Via a Primal-Dual Approach},
  booktitle    = {47th Annual {IEEE} Symposium on Foundations of Computer Science, {FOCS}},
  pages        = {293--304},
  year         = {2006},
  url          = {https://doi.org/10.1109/FOCS.2006.39},
  doi          = {10.1109/FOCS.2006.39}
}

@inproceedings{onlinebipartitite,
  author       = {Richard M. Karp and
                  Umesh V. Vazirani and
                  Vijay V. Vazirani},
  title        = {An Optimal Algorithm for On-line Bipartite Matching},
  booktitle    = {Proceedings of the 22nd Annual {ACM} Symposium on Theory of Computing, {STOC}},
  pages        = {352--358},
  year         = {1990},
  url          = {https://doi.org/10.1145/100216.100262},
  doi          = {10.1145/100216.100262}
}

@article{onlinepaging,
  author       = {Daniel Dominic Sleator and
                  Robert Endre Tarjan},
  title        = {Amortized Efficiency of List Update and Paging Rules},
  journal      = {Commun. {ACM}},
  volume       = {28},
  number       = {2},
  pages        = {202--208},
  year         = {1985},
  url          = {https://doi.org/10.1145/2786.2793},
  doi          = {10.1145/2786.2793},
  bibsource    = {dblp computer science bibliography, https://dblp.org}
}

@inproceedings{onlinebipartititeextended,
  author       = {Nikhil R. Devanur and
                  Kamal Jain and
                  Robert D. Kleinberg},
  title        = {Randomized Primal-Dual analysis of {RANKING} for Online BiPartite Matching},
  booktitle    = {Proceedings of the Twenty-Fourth Annual {ACM-SIAM} Symposium on Discrete Algorithms, {SODA}},
  pages        = {101--107},
  year         = {2013},
  url          = {https://doi.org/10.1137/1.9781611973105.7},
  doi          = {10.1137/1.9781611973105.7}
}

@inproceedings{AzarBCCCG0KNNP16,
  author       = {Yossi Azar and
                  Niv Buchbinder and
                  T.{-}H. Hubert Chan and
                  Shahar Chen and
                  Ilan Reuven Cohen and
                  Anupam Gupta and
                  Zhiyi Huang and
                  Ning Kang and
                  Viswanath Nagarajan and
                  Joseph Naor and
                  Debmalya Panigrahi},
  title        = {Online Algorithms for Covering and Packing Problems with Convex Objectives},
  booktitle    = {{IEEE} 57th Annual Symposium on Foundations of Computer Science, {FOCS}},
  pages        = {148--157},
  year         = {2016},
  url          = {https://doi.org/10.1109/FOCS.2016.24},
  doi          = {10.1109/FOCS.2016.24}
}

@article{BansalBN12,
  author       = {Nikhil Bansal and
                  Niv Buchbinder and
                  Joseph Naor},
  title        = {A Primal-Dual Randomized Algorithm for Weighted Paging},
  journal      = {J. {ACM}},
  volume       = {59},
  number       = {4},
  pages        = {19:1--19:24},
  year         = {2012},
  url          = {https://doi.org/10.1145/2339123.2339126},
  doi          = {10.1145/2339123.2339126},
  bibsource    = {dblp computer science bibliography, https://dblp.org}
}

@inproceedings{CoesterTT26,
  author       = {Christian Coester and
                  Alexa Tudose and
                  Alexander Turoczy},
  title        = {Learning-Augmented Online Minimization with Dual Predictions},
  booktitle    = {Forty-third International Conference on Machine Learning, {ICML}},
  year         = {2026}
}

\appendix

\section{Omitted Proofs} \label{app:omittedproofs}

\lemrecurrenceproperty*
\begin{proof}
Let
\[
h(x):=\frac{x}{e^x-1},
\qquad
\Phi(x):=\int_1^x \frac{dt}{h(t)}
=\int_1^x \frac{e^t-1}{t}\,dt .
\]
Thus the recurrence is \(S_{K+1}=S_K+h(S_K)\).

First note that \(S_K\) is increasing. It cannot converge to a finite limit \(L\), since then
\(h(S_K)\to h(L)>0\), contradicting convergence. Hence $S_K\to\infty$.

Since $h(x)\to 0$ as $x\to\infty$, we further see that $S_K=o(K)$.

We now compare the discrete recurrence with the integral \(\Phi\). Since
\[
\Phi'(x)=\frac{1}{h(x)}=\frac{e^x-1}{x}
\]
and
\[
\Phi''(x)
=
\frac{(x-1)e^x+1}{x^2}
=
O\!\left(\frac{e^x}{x}\right),
\]
Taylor's formula gives, for \(x\to\infty\),
\[
\Phi(x+h(x))-\Phi(x)
=
\Phi'(x)h(x)
+
O\!\left(\frac{e^x}{x}h(x)^2\right).
\]
But \(\Phi'(x)h(x)=1\), and
\[
\frac{e^x}{x}h(x)^2
=
\frac{e^x}{x}\left(\frac{x}{e^x-1}\right)^2
=
O(xe^{-x})
=
O(h(x)).
\]
Therefore
\[
\Phi(x+h(x))-\Phi(x)=1+O(h(x)).
\]
Applying this with \(x=S_j\) and summing from \(j=1\) to \(K-1\), we obtain
\[
\Phi(S_K)
=
K-1+O\!\left(\sum_{j=1}^{K-1} h(S_j)\right)
=
K-1+O(S_K)
=
K+o(K).
\]
Hence $\Phi(S_K)\sim K$.

It remains to estimate \(\Phi\). By l'H\^opital's rule,
\[
\lim_{x\to\infty}\frac{\Phi(x)}{e^x/x}
=
\lim_{x\to\infty}
\frac{(e^x-1)/x}{e^x(x-1)/x^2}
=
1.
\]
Thus $\Phi(x)\sim \frac{e^x}{x}$

Combining this with \(\Phi(S_K)\sim K\), we get
\[
K\sim \frac{e^{S_K}}{S_K}.
\]
Taking logarithms yields
\begin{align}
\ln K
=
S_K-\ln S_K+o(1).\label{eq:lnK=SK-lnSK+o1}
\end{align}
Since \(S_K\to\infty\), we have \(\ln S_K=o(S_K)\), and therefore equation~\eqref{eq:lnK=SK-lnSK+o1} implies
\[
\frac{S_K}{\ln K}\to 1.
\]
Consequently,
\[
\ln S_K
=
\ln\ln K+o(1).
\]
Substituting this into equation~\eqref{eq:lnK=SK-lnSK+o1} proves the lemma.
\end{proof}
% \begin{proof}
%     We provide an informal derivation of $S_K = \ln K + \ln \ln K + o(1)$, because a fully rigorous derivation of the asymptotics would be a distraction from our main results.
%     We can rewrite the recurrence as $S_{K} - S_{K-1} = \frac{1}{\exp(S_{K-1}) - 1} \approx \frac{1}{\exp(S_{K-1}-1)}$. Once $S_K$ is large, we can approximate the discrete recurrence $S_K$ by a continuous function $S(K)$ with an ODE given by
%     \begin{equation}
%         \frac{dS}{dK} = \frac{1}{\exp(S)-1} \approx \frac{1}{\exp(S)}
%     \end{equation}

%     Meaning $K \sim \int \frac{1}{\exp(S)}dS= \text{Ei}(S)$, where $\text{Ei}$ is the exponential integral. It is known from asymptotic analysis texts \cite{asymptoticsreference} that $\text{Ei}(S) \sim \frac{e^S}{S}$, and hence $K \sim \frac{e^S}{S}$. This means $\ln K \sim S - \ln S$. 
    
%     We arrive at $S - \ln S = \ln (K) + o(1)$. By iteratively bootstrapping $S$ with better approximations, one obtains $S = \ln(K) + \ln(\ln(K))+ o(1) $.
% \end{proof}

\end{document}